\documentclass[conference,a4paper]{IEEEtran}

\usepackage[utf8]{inputenc} 
\usepackage[T1]{fontenc}
\usepackage{url}
\usepackage{ifthen}
\usepackage[noadjust]{cite}
\usepackage[cmex10]{amsmath} 
\usepackage{amsmath}
\usepackage{amsthm}
\usepackage{amssymb}
\usepackage{algorithm}
\usepackage{algorithmicx}
\usepackage{algpseudocode}
\usepackage{bm}
\usepackage{epstopdf}
\usepackage{xcolor}
\usepackage{graphicx}

\newtheorem{theorem}{Theorem}[section]
\newtheorem{lemma}[theorem]{Lemma}
\newtheorem{definition}[theorem]{Definition}

\newtheorem{corollary}[theorem]{Corollary}

\newcommand{\N}{\mathcal{N}}

\DeclareMathAlphabet{\mathbbold}{U}{bbold}{m}{n}

\begin{document}
\title{From Centralized to Decentralized Coded Caching} 


\author{%
  \IEEEauthorblockN{Yitao~Chen, Alexandros G.~Dimakis}
  \IEEEauthorblockA{University of Texas at Austin\\
                    Austin, TX 78701 USA\\
                    Email: yitaochen@utexas.edu, dimakis@austin.utexas.edu}
  \and
  \IEEEauthorblockN{Karthikeyan Shanmugam}
  \IEEEauthorblockA{IBM Research AI\\
                    Yorktown Heights, NY 10504 USA  \\ 
                    Email: karthikeyanshanmugam88@gmail.com}
}


\maketitle

\begin{abstract}
We consider the problem of designing decentralized schemes for coded caching. In this problem there are $K$ users each caching $M$ files out of 
a library of $N$ total files. The question is to minimize $R$, the number of broadcast transmissions to satisfy all the user demands. Decentralized schemes allow the creation of each cache independently, allowing users to join or leave without dependencies. Previous work showed that to achieve a coding gain $g$, i.e. $R \leq K (1-M/N)/g$ transmissions, each file has to be divided into number of subpackets that is exponential in $g$. 

In this work we propose a simple translation scheme that converts any constant rate centralized scheme into a random decentralized placement scheme that guarantees a target coding gain of $g$. If the file size in the original constant rate centralized scheme is subexponential in $K$, then the file size for the resulting scheme is subexponential in $g$. When new users join, the rest of the system remains the same. However, we require an additional communication overhead of $O(\log K)$ bits to determine the new user's cache state. We also show that the worst-case rate guarantee degrades only by a constant factor due to the dynamics of user arrival and departure.
 \end{abstract}

\textit{A full version of this paper is accessible at: \url{https://yitaochen.github.io/files/codedcache_isit18.pdf}}

\section{Introduction}
 Demand for wireless bandwidth has increased dramatically owing to rise in mobile video traffic \cite{mobile2011global, boccardi2014five}. One of the most promising approaches for design of next generation networks (5G)  is to densify deployment of small/micro/femto cell stations. One main issue is that the backhaul networks required for such a dense deployment is a severe bottleneck. To alleviate this, a vast number of recent works proposed caching highly popular content at users and or at femto cell stations near users \cite{shanmugam2013femtocaching, bastug2014living}. These caches could be populated during off-peak time periods by predictive analytics. This caching at the `wireless edge' is being seen as a fundamental component of 5G networks \cite{boccardi2014five, paschos2016wireless}.
 
  Upon a cache hit, users obtain the files using near-field communication from nearby femto stations or directly retrieve it from their local caches. Another non-trivial benefit is the possibility of coded transmissions leveraging cache content. One or more packets can be XORed by a macro base station and sent. Users can decode the required packets by using local cache content. Potentially, the benefit over and above that obtained only through cache hits can be enormous. A stylized abstract problem that explores this dimension is called the coded caching problem, introduced by Maddah-Ali and Niesen in their pioneering work \cite{maddah2014fundamental}.
  
  In the coded caching problem, $K$ users are managed by a single server through a noiseless broadcast link. Each user demand arises from a library of $N$ files. Each user has a cache memory of $M$ files. Each file consists of $F$ subpackets. There are two phases - a placement and delivery phase. In the placement phase, every user cache is populated by packets of different files from the library. In the delivery phase, user demands are revealed (the choice could be adversarial). The broadcast agent sends a set of coded packets such that each user can decode its desired file using its cache content designed from the placement phase. The objective is to jointly design both phases such that the worst-case number of file transmissions (often called as the \textit{rate}) is at most $R$. The most surprising result is that $R \leq N/M$, that is independent of the number of users can be achieved. This was shown to be information theoretically optimal upto constant factors. There has been a lot of work \cite{maddah2013decentralized,niesen2013coded,ji2017order,ji2014multiple,karamchandani2016hierarchical,hassanzadeh2017rate} extending this order optimal result to various settings - demands arising from a popularity distribution, caching happening at various levels etc.
  
  There is another line of work that focuses on minimizing file size - the number of subpackets $F$ required - for a given worst-case rate. There have been two types of coded caching schemes - a) Centralized and b) Decentralized schemes. Centralized schemes have deterministic and coordinated placement and delivery phases. Specifically, if an additional user arrives, all caches have to be reconfigured. Decentralized schemes have a random placement phase and the objective is to optimize the worst-case rate with high probability over the randomization in the placement phase. For all known decentralized schemes, the random cache content of a new user is independent of the rest of the system. This removes the need for system wide changes when new users arrive and leave the system.
  
   Initial centralized schemes required file sizes exponential in $K$ to obtain constant worst case rate (we always assume ratio $M/N$ is a constant in this work that does not scale with $K$). Subsequent works \cite{Yan2017,yan2016placement,cheng2017optimal,cheng2017coded,shangguan2016centralized,tang2017low} have explored centralized schemes that attain sub-exponential file size and constant worst-case rate. Even linear file size for near-constant rates is feasible in theory although this requires impractically large values of $K$ \cite{ShanmugamISIT}. The original decentralized schemes required exponential file size in $K$ even for a constant \textit{coding gain} of $g$, i.e. $R \leq K/g$ w.h.p \cite{ShanmugamIT}. This was the price required for decentralization in the initial scheme. Subsequent works have reduced the file size to exponentially depend on only $g$ (the target coding gain) independent of the number of users $K$ \cite{ShanmugamIT, ji2015efficient, ji2015efficient2,jin2017structural}. However, there are no decentralized schemes known (as far as the authors are aware) that have file size F scaling subexponentially in the target coding gain $g$. 
  
 \textbf{Our Contributions:} In this work, inspired by results in \cite{jin2017structural} and leveraging ideas from balls and bins literature with power of two choices, we show the following:
  \begin{enumerate}
   \item We provide a simple translation scheme that takes any centralized scheme with constant rate and subexponential file size scaling with the number of users and turns into a decentralized scheme with target coding gain $g$ with file size that is subexponential in $g$. This generic translation scheme when applied to a known centralized scheme gives a feasible decentralized scheme whose file size is subexponential in $\sqrt{g}$.
   \item Our decentralized scheme does not require any change in the rest of the system when a new user joins. However, it requires an additional $O(\log K)$ bits of communication between the server and a newly joining user. We also show that the worst case rate degrades by at most a constant factor when there are not too many adversarial arrivals and departures.
      \item Finally, we show that the centralized scheme with near constant rates and polynomial file size requirements can also be translated into decentralized schemes that provide a polynomial scaling in the target gain $g$.
  \end{enumerate}
  
  In summary, we show that good centralized schemes can be mapped to decentralized schemes with similar performance.  We emphasize that our decentralized schemes are not fully independent (as opposed to all previous decentralized methods), but still allow users to easily join or leave the system.
\section{Problem Setting}
\subsection{Coded Caching Problem}
In this part, we formally define the coded caching problem. Consider $K$ users that request files from a library of size $N$. We are mostly interested in the case when $K < N$. The $N$ files are denoted by $W_1, \ldots, W_N$, consisting of $F$ data packets. Each file packet belongs to a finite alphabet $\chi$. Let $\N  \triangleq \{1,2,\ldots,N\}$ and $\mathcal{K} \triangleq \{1,2,\ldots, K\}$ denote the set of files and the set of users, respectively. Each user has a cache that can store $MF$ packets from the library, $M \in [0, N]$. In the \emph{placement phase}, user caches are populated without knowledge of the user demands. Let $\phi_u$ denote the caching function for user $u$, which maps $N$ files $W_1,\ldots, W_N$ into the cache content $Z_u \triangleq \phi_u(W_1,\ldots, W_N)\in \chi^{MF}$ for user $u \in \mathcal{K}$. Let $\bm{Z} \triangleq (Z_1,\ldots,Z_K)$ denote the cache contents of all the $K$ users. In the \emph{delivery phase}, where users reveal their individual demands $\bm{d} \triangleq (d_1,\ldots, d_K) \in \N^K$, let $\psi$ denote the encoding function for the server, which maps the files $W_1,\ldots, W_N$, the cache contents $\bm{Z}$, and the request $\bm{d}$ into the multicast message $Y \triangleq \psi(W_1,\ldots,W_N,\bm{Z},\bm{d})$
sent by the server over the shared link. Let $\gamma_u$ denote the decoding function at user $u$, which maps the multicast message $Y$, the cache content $Z_u$ and the request $d_u$, to estimate $\hat{W}_{d_u} \triangleq \gamma_u(Y,Z_u,d_u)$ of the requested file $W_{d_u}$ of user $u \in \mathcal{K}$. Each user should be able to recover its requested file from the message received over the shared link and its cache content. Thus, we impose the successful content delivery condition 
\begin{flalign}\label{equ:succDec}
\hat{W}_{d_u}=W_{d_u}, \quad \forall u \in \mathcal{K}.
\end{flalign}
Given the cache size ratio $M/N$, the cache contents $\bm{Z}$ and the requests $\bm{d}$ of all the $K$ users, let $R(M/N, K, \bm{Z},\bm{d})F$ be the length of the multicast message $Y$. Let 
\begin{flalign*}
R(M/N, K, \bm{Z}) \triangleq \max_{\bm{d} \in \N^K} R(M/N, K, \bm{Z},\bm{d})
\end{flalign*}
denote the worst-case (normalized) file transmissions over the shared link. 
The objective of the coded caching problem is to minimize the worst-case file transmissions $R(M/N, K, \bm{Z})$. The minimization is with respect to the caching functions $\{\phi_k: k \in \mathcal{K}\}$, the encoding function $\psi$, and the decoding functions $\{\gamma_k: k \in \mathcal{K}\}$, subject to the successful content delivery condition in (\ref{equ:succDec}). A set of feasible placement and delivery strategies constitutes a coded caching scheme.

\subsection{Two types of schemes}
As we state in the introduction, there are two types of coded caching schemes - a) Centralized Schemes and b) Decentralized Schemes. Now we further divide the decentralized schemes into two kinds in this work for the purpose of illustrating our results in contrast to existing ones.
\begin{enumerate}
\item \textit{Decentralized Type A} The random set of file packets placed in any user $u$'s cache is independent of the rest of the system requiring no coordination in the placement phase when users join the system and leave. Most of the current known (as far as the authors are aware) decentralized schemes are of this kind. 
\item \textit{Decentralized Type B} When a new user $u$ joins, the random set of file packets placed in any users $u$'s cache is dependent of the rest of the system. However, it does not require any change in the rest of the system. We also seek to minimize the number of bits $B$ communicated when the new user's cache state is determined.
\end{enumerate}

\subsection{Objective}
The prime focus in this work is to design Decentralized Schemes of type B such that for a given worst-case rate (with high probability \footnote{We say an event $A$ occurs with high probability (w.h.p.) if $\Pr(A) \ge 1- O(g^{-\alpha})$ for a constant $\alpha \ge 1$.} with respect to the random placement scheme) of at most $K(1-M/N)/g$, for constant $M/N$, the file size $F$, as a function of the coding gain $g$, is kept small as possible. The number of bits communicated $B$ when users join and leave the system also needs to be minimized.


\section{Preliminary}

\subsection{Centralized Schemes - Ruzsa-Szemer\'edi constructions}
In this section, we introduce a class of centralized coded caching schemes called Ruzsa-Szemer\'edi schemes. We describe a specific family of bipartite graphs call {\em Ruzsa Szemer\'edi} bipartite graphs. Then, we review an existing connection between these bipartite graphs and centralized coded caching schemes. 

\begin{definition}
Consider an undirected graph $G(V,E)$. An induced matching $M \subseteq E$ is a set of edges such that a) no two edges in $M$ share a common vertex and b) the subgraph induced by the vertices in the matching contains only the edges in $M$ and no other edge in the original graph $G$.
\end{definition}

\begin{definition}
A bipartite graph $G([F],[K],E)$ is an $(r,t)$-Ruzsa-Szemeredi graph if the edge set can be partitioned into $t$ induced matchings and the average size of these induced matchings is $r$.
\end{definition}

Now, we describe a coded caching scheme-placement and delivery phases-from the construction of a Ruzsa-Szemeredi bipartite graph. 

\begin{theorem}\cite{ShanmugamISIT}
Consider a Ruzsa-Szemer\'edi bipartite graph on vertex sets $[F]$ and $[K]$ such that the minimum right-degree is $c \le F$. Then, for any $M/N \ge 1- c/F$, we have a centralized coded caching scheme with worst case rate $R=t/F$ with system parameters $(K,M,N,F)$.
\end{theorem}
With a given $(K,M,N,F)$ Ruzsa-Szemer\'edi bipartite graph $G([F],[K],E)$, an $F$-packet coded caching scheme can be realized by Algorithm 1. In the placement phase, non-edge represents storage actions.  An edge $e \in E$ between $f \in F$ and $k \in K$ is denoted by $(f,k)$. If $(f,k) \notin E$, then file packet $f$ of all files is stored in user $k$'s cache. In the delivery phase, an XOR of all the packets involved in an induced matching is sent. We repeat this XORing process for every induced matching. This policy yields a feasible delivery scheme that satisfies any demand set $\bm{d}$. 


Almost all (as far as the authors are aware) known centralized coded caching schemes belong to the class of Ruzsa-Szemer\'edi schemes. They have been introduced in the literature through several other equivalent formulations (like placement delivery array etc.)\cite{Yan2017,yan2016placement,cheng2017optimal,cheng2017coded,shangguan2016centralized,tang2017low}. 
\begin{algorithm}                               
\caption{Ruzsa-Szemeredi based caching scheme}          
\label{alg:0}                                    
  \begin{algorithmic} 
 	  \Procedure{ Placement}{$G([F],[K],E),\{W_n, n \in \mathcal{N}\}$}
	  	\State Split each file $W_n, n \in \mathcal{N}$ into $F$ packets, i.e., $W_n=\{W_{n,f}: f=1,2,\cdots, F\}$
		\For{ $k\in [1:K]$}
			\State $Z_k \leftarrow \{W_{n,f}: (f,k) \notin E, \forall n = 1,2,\cdots, N \}$
		\EndFor
	  \EndProcedure
	  \Procedure{ Delivery}{$G([F],[K],E),\{W_n, n \in \mathcal{N}\},\bm{d}$}
	  	\For{$s=1,2,\cdots,t$}
			\State Suppose $(f_1,k_1),\ldots, (f_p,k_p)$ represents a $p$-sized induced matching. 
			\State Server sends $\oplus_{j\in [p]} W_{d_{k_j}, f_j}$
		\EndFor
	  \EndProcedure
  \end{algorithmic}
\end{algorithm}
%
In the next section, we define a new `translation' mechanism that generates a decentralized scheme of type B out of an existing class of Ruzsa-Szemer\'edi schemes of constant rate that preserves the efficiency of file size requirements. 

\section{Our Decentralized Scheme}

\subsection{Translation using Balls and Bins Argument}\label{decentsubsec}
Our objective is to specify a decentralized scheme for $K$ users, system parameters $M$ and $N$ and a worst-case rate of at most $K(1-M/N)/g$ w.h.p. First, given the target coding gain $g$, the size of cache memory $M$, the number of files $N$ and the number of users $K$, we decide an appropriate number of virtual users $K'$. We assume that we can construct Ruzsa-Szemer\'edi centralized schemes for $K'=f(M,N,K,g)$ (this function will be specified later), for constant $M/N$ and worst-case rate $R$ which is dependent only on $M$ and $N$ and file size requirement $F$. Consider the cache content of every virtual user $k$ according to this centralized scheme. Let us denote the virtual user's cache content by $C_k \in \chi^{MF}$. Please note that the cache contents of the centralized scheme is only virtual. We specify the random placement scheme for $K$ real users as follows.

\textbf{Placement Scheme:}
For each real user $u \in [1:K]$ in sequence, we pick two virtual cache contents $C_{u_1}$ and $C_{u_2}$ at random. We assign the cache content of the real user $u$ to that virtual cache content which has been least used so far amongst $C_{u_1}$ and $C_{u_2}$. Let us denote by $X_k$ the number of real users which store $C_k$. 

\textbf{Balls and Bins:} We specify a one-one correspondence to a balls and bins system. The number of distinct virtual user cache contents are the bins in the system. There are $K'$ of them. A real user corresponds to a ball. When a ball in placed in the bin, a real user (ball) is assigned the cache content of that virtual user (bin it is placed in). We can easily see that the random placement exactly corresponds to a power of two choices in a standard balls and bin process \cite{azar1999balanced,cole1998balls}.

\textbf{Delivery Scheme:}
 Note that, in a system with $K'' <K'$ users with distinct cache contents $C_1 \ldots C_{K''}$, by using the $(K',M,N)$ Ruzsa-Szemer\'edi delivery scheme with files demanded by users $k> K''$ substituted by a dummy file, it is possible to still guarantee a worst-case rate of $R$ in the delivery phase.
 
We now repeatedly perform the following until all $X_k=0$: Find a set of at most $K$ real users with maximum number of distinct virtual cache contents. Subtract $X_k$ corresponding to those virtual cache contents by $1$. Use the $(K',M,N)$ Ruzsa-Szemer\'edi delivery scheme for these real users and their real demands. Clearly, the total number of worst case transmissions is at most $R* \max_k X_k$. We summarize the decentralized scheme in Algorithm \ref{alg:1}.

\begin{algorithm}[!hbp]                           
\caption{Decentralized Scheme}          
\label{alg:1}                                    
  \begin{algorithmic}
      \State Given $M,N,K,g$, let $K'=f(g,M/N,K)$ (depends on constructions).
      \State Get the cache contents $C_1,C_2 \ldots C_{K'}$ corresponding to the Ruzsa-Szemer\'edi placement scheme (in Algorithm \ref{alg:0}) with parameters $(K',M,N,F)$.
       \Procedure{Sampling}{$\mathcal{C}, K,\{X_k\},u$}      	    
      		\State Uniformly sample a cache content from $\{C_1 \ldots C_{K'}\}$ for the cache of user $u$ twice with replacement, i.e., $C_{u_1},C_{u_2}$.
			\If{$X_{u_1} \leq X_{u_2}$}
				\State $X_{u_1} \leftarrow X_{u_1} + 1$. $Z_u \leftarrow C_{u_1}$
                        \Else
				\State $X_{u_2} \leftarrow X_{u_2} + 1$. $Z_u \leftarrow C_{u_2}$
			\EndIf
      \EndProcedure
      \Procedure{ Placement}{$M,N,K,g$}
           \State Initialize $X_k=0, ~\forall k \in [1:K']$.
          \For{$u=1,2,\cdots, K$}
	   \State Sampling($\mathcal{C}, K,\{X_k\},u$)
	   \EndFor
       \EndProcedure
      \Procedure{Delivery}{$M,N,K,g$}
      \State Let $S_k \leftarrow X_k$. Let ${\cal K} \leftarrow \{1 \ldots K \}$.
      \While {$\max_k S_k >0$}
          \State Find a maximal subset $R \subset {\cal K}$ such that cache contents of all real users assigned in $F$ are distinct.
          \If {$\lvert R \rvert < K'$}
             \State Use Delivery subroutine of Algorithm \ref{alg:0} to satisfy demands of users in $R$ using a $(K',M,N)$ Ruzsa-Szemer\'edi Scheme. This can be done by substituting packets belonging to file demands of users outside set $R$ (since $\lvert R \rvert <K'$) by packets from a dummy file known to all users.         
          \Else
            \State Use Delivery subroutine of Algorithm \ref{alg:0} to satisfy demands of users in $R$ ($\lvert R \rvert=K'$) using a $(K',M,N)$ Ruzsa-Szemer\'edi Scheme.
          \EndIf
      \EndWhile
      \EndProcedure 
  \end{algorithmic}
\end{algorithm}


\subsection{Analysis of the decentralized algorithm}\label{analysisstatic}

\begin{lemma} \label{lem:rate}
The total number of worst-case file transmissions of the delivery scheme in Algorithm \ref{alg:1} is given by:
\begin{flalign*}
R(g,M/N,K,\{X_1 \ldots X_{K'}\})= R* \max_k X_k
\end{flalign*}
where $R$ is the worst-case rate of the $(K',M,N)$ Ruzsa-Szemer\'edi of Algorithm \ref{alg:1}.
\end{lemma}
\begin{proof}
 The delivery scheme of Algorithm $1$ is called, with possibly dummy user demands, at most $\max_k X_k$ times. Each call produces at most $R$ file transmissions. The proof follows from this.
\end{proof}


As we stated before, the placement has a direct correspondence to a choice of two balls and bins process. There are $K$ balls and $K'$ bins. In sequence, for every ball, two bins are chosen uniformly randomly with replacement and the ball is placed in the bin with least number of balls. From \cite{berenbrink2000balanced}, we have the following Lemma:
\begin{lemma}[\cite{berenbrink2000balanced}]\label{lem:lnlnk}
The maximum number of balls in any bin, achieved by the choice of two policy for balls and bins problem, with $K$ balls and $K'$ bins, $K \ge K'$ is less than $K/K'+\ln \ln K'/\ln 2+9$ with probability at least $1-O((K')^{-\alpha})$, where $\alpha \ge 1$ is a suitable constant.
\end{lemma}

From Lemma \ref{lem:lnlnk}, we know that $X_{\max} \leq K/K'+\ln \ln K'/\ln 2+O(1)$ probability at least $1- \frac{1}{(K')^{\alpha}}$ for some constant $\alpha$. Therefore, have the following theorem:

\begin{theorem}\label{thm:meta}
Suppose there exists a Ruzsa-Szemer\'edi centralized scheme, with constant (independent of $K'$) worst-case rate $R_c$, constant cache size ratio $\frac{M}{N}$, and subpacketization level $F_c=O(2^{(K')^{\delta}f(R_c,M/N)})$. Then the scheme in Algorithm \ref{alg:1} that uses this centralized scheme has target gain $g$, i.e. the number of file transmissions in the worst case is rate $R_d=\frac{K(1-M/N)}{g}+O(\ln \ln g)$  w.h.p. The subpacketization level required is $F_d=O(2^{g^{\delta} h(R_c,M/N)})$  where $h(R_c,M/N)=f(R_c,M/N)R_c^{\delta}/(1-M/N)^{\delta}$. To obtain the scheme, we set $K'=gR_c/(1-M/N)$ in Algorithm \ref{alg:1}.
 \end{theorem}
\begin{proof}
The proof can be found in the full version.
\end{proof}

From \cite{shangguan2016centralized}, we have the following lemma,
\begin{lemma}[\cite{shangguan2016centralized}]\label{lem:hype}
There exists an $(r,t)$-Ruzsa Szemer\'edi graph $G([F], [K], E)$ with $t=\binom{n}{a+2}, F=\binom{n}{a}, K=\binom{n}{2}$ for some $a$ and $n=\lambda a$ for some constant $\lambda>1$. Then, $R=t/F=\binom{n}{a+2}/\binom{n}{a}\approx (\lambda-1)^2$ and $M/N=(\binom{n}{a}-\binom{n-2}{a})/\binom{n}{a}\approx\frac{2\lambda-1}{\lambda^2}$ and by Stirling's formula we have:
\begin{flalign*}
F=\binom{n}{\lambda^{-1}n}&=\frac{1+o(1)}{\sqrt{2\pi \lambda^{-1}(1-\lambda^{-1})n}}\cdot 2^{nH(\lambda^{-1})}\\
&=O(K^{-1/4}\cdot 2^{\sqrt{2K} H(\lambda^{-1})}),
\end{flalign*}
where $H(x)=-x\log_2 x-(1-x)\log_2(1-x)$ for $0 < x <1$ is the binary entropy function. It is easy to see that under such choice of parameters, $R$ and $M/N$ are both constants independent of $K$ and $F$ grows sub-exponentially with $K$.
\end{lemma}
Apply Theorem \ref{thm:meta} with Lemma \ref{lem:hype}, we have the following Corollary,
\begin{corollary}
For the Ruzsa Szemer\'edi centralized scheme in Lemma \ref{lem:hype}, 
we have a corresponding decentralized scheme with $R_d=\frac{K(1-M/N)}{g}+O(\ln \ln g)$ and subpacketization level $F=2^{\tilde{O}(g^{1/2})}$, where $\tilde{O}$ means $\log g$, $\mathrm{poly}(\lambda)$ terms are omitted.
\end{corollary}
For non-constant rate and constant cache size ratio $M/N$, we have the following theorem:
\begin{theorem}\label{thm:kdelta}
Suppose there exists a Ruzsa-Szemer\'edi centralized scheme, with rate $R_c \leq (K')^{\delta}$, $\delta \in (0,1)$ and constant cache size ratio $M/N$, and subpacketization level $F_c=p(K')$ where $p(K')$ is a polynomial in $K'$. Then the scheme in Algorithm~\ref{alg:1} that uses this centralized scheme has target gain $g$, i.e., the number of file transmissions in the worst case is rate $R_d=\frac{K(1-M/N)}{g}+O(g^{\delta/(1-\delta)}\ln \ln g)$ w.h.p. The subpacketization level required is $F_d=p \left((\frac{g}{(1-M/N)})^{1/(1-\delta)} \right)$. To obtain the scheme, we set $K'=(\frac{g}{1-M/N})^{1/(1-\delta)}$.
\end{theorem}
\begin{proof}
The proof can be found in the full version.
\end{proof}

\subsection{Overhead analysis for the dynamic version of the decentralized scheme}
 We consider the dynamics of user arrival and departure. 
 When a user leaves the system, the user's cache content is deleted and if the user had cached $C_k$ (Recall from Section \ref{decentsubsec}, that this is the cache content of the $k$-th virtual user from Section \ref{decentsubsec}), $X_k$ is decreased by $1$. When a new user $u$ joins the system, then the subroutine $\mathrm{Sampling}\left(\mathcal{C}, K,\{X_k\},u\right)$ from Algorithm \ref{alg:1} is executed to determine the cache content of user $u$ (i.e. $Z_u$). The comparison between $X_{u_1}$ and $X_{u_2}$ in the procedure $\mathrm{Sampling}(\cdot)$ involves an additional $3 \log K$ bits of communication overhead between user $u$ and the central server. Note that, the dynamics of user arrivals and departure does not change the cache contents of users already in the system.
 
 The worst-case rate during delivery is directly proportional to $\max_k X_k$ according to Lemma \ref{lem:rate}. We show that, despite the dynamics, $\max_k X_k$ remains the same upto constant factors w.h.p provided the number of adversarial departures and arrivals is bounded. We recall that the real users represent the balls and the virtual users or their distinct cache contents $C_i$ represent the bins. $\max_k X_k$ is the size of the maximum bin. 

For the analysis, let us first define the balls and bins process with adversarial deletions/additions. Consider the polynomial time process where in the first $K$ steps, a new ball is inserted into the system (the system is initiated with $K$ users). At each subsequent time step, either a ball is removed or a new ball is inserted in the system, provided that the number of balls present in the system never exceeds $K$. 
 Suppose that an adversary specifies the full sequence of insertions and deletions of balls in advance, without knowledge of the random choice of the new balls that will be inserted in the system (i.e., suppose we have an oblivious adversary).

Our proof uses Theorem 1 from \cite{cole1998balls} and Theorem 3.7 from \cite{azar1999balanced} to obtain performance guarantees as dynamic users join and leave the system. 
\begin{theorem}\label{thm:dc}
For any fixed constant $c_1$ and $c_2$ such that $(K')^{c_2} > K$, if the balls and bins process with adversarial deletions runs for at most $(K')^{c_2}$ times steps, then the maximum load of a bin during the process is at most $O(L/K')+\ln \ln K'/\ln 2+O(c_1+c_2)$, with probability at least $1-o(1/(K')^{c_1})$.
\end{theorem}
\begin{proof}
Please refer to the full version for a self-contained proof that extends results from previous work.
\end{proof}


\section{Conclusion}
In this work we show a simple translation scheme that converts any constant rate centralized scheme into a random decentralized placement scheme that guarantees a target coding gain of $g$. We show the worst-case rate due to the dynamics of user arrival and departure degrades only by a constant factor. 

\bibliographystyle{IEEEtran}
\bibliography{refruzsa}

\newpage
\appendix

\subsection{Proof of Theorem \ref{thm:meta}}
Consider a Ruzsa-Szemer\'edi centralized scheme with constant $R_c$, according to Lemma \ref{lem:rate} and \ref{lem:lnlnk}, the rate in the corresponding decentralized scheme is 
\begin{flalign*}
R_d&= R_c \max_{k}X_k \\
&\le R_c \left(\frac{K}{K'} + \frac{\ln \ln K'}{\ln 2}+O(1) \right)\\
&= \frac{K(1-M/N)}{g} + O(\ln \ln g).
\end{flalign*}
Here, we set $K'=gR_c/(1-M/N)$. Substituting $K'=gR_c/(1-M/N)$ into the number of subpackets required in the centralized scheme $F_c$, then that required for Algorithm \ref{alg:1} is:
\begin{flalign*}
F_d=O(2^{g^{\delta} h(R_c,M/N)}),
\end{flalign*}
where $h(R_c,M/N)=f(R_c,M/N)R_c^{\delta}/(1-M/N)^{\delta}$.
\subsection{Proof of Theorem \ref{thm:kdelta}}
Consider a Ruzsa-Szemer\'edi centralized scheme with rate $R_c=(K')^{\delta}$, $\delta \in (0,1)$, according to Lemma \ref{lem:rate} and \ref{lem:lnlnk}, the rate in the corresponding decentralized scheme is 
\begin{flalign*}
R_d&= R_c \max_{k}X_k \\
&\le R_c \left(\frac{K}{K'} + \frac{\ln \ln K'}{\ln 2} + O(1) \right)\\
&= \frac{K}{(K')^{1-\delta}}+(K')^{\delta} \left( \frac{\ln \ln K'}{\ln 2} + O(1)\right) \\
&= \frac{K(1-M/N)}{g} + O(g^{\delta/(1-\delta)}\ln \ln g).
\end{flalign*}
Here, we set $K'=(g/(1-M/N))^{1/(1-\delta)}$. Substituting $K'=(g/(1-M/N))^{1/(1-\delta)}$ into the number of subpackets required in the centralized scheme, the number of subpackets required for Algorithm \ref{alg:1} is:
\begin{flalign*}
F_d=O(\text{p}(g^{1/(1-\delta)})).
\end{flalign*}

\subsection{Proof of Theorem \ref{thm:dc}} 

For a vector $\bm{v}=(v_1,v_2,\ldots)$, let $P_2(\bm{v})$ be the following process: at time steps 1 through $K$, $K$ balls are placed into $K'$ bins sequentially, with each ball going into the least loaded of $2$ bins chosen independently and uniformly at random. After these balls are placed, deletions and insertions alternate, so that at each subsequent time step $K+j$, first the ball inserted at time $v_j$ is removed, and then a new ball is placed into the least loaded of $2$ bins chosen independently and uniformly at random. (Actually we do not require this alternation; the main point is that we have a bound, $K$, on the number of balls in the system at any point. The alternation merely makes notation more convenient.)

We assume the vector $\bm{v}$ is suitably defined so that at each step an actual deletion occurs; that is, the $v_j$ are unique and $v_j \le K+j-1$. Otherwise $\bm{v}$ is arbitrary, although we emphasize that it is chosen before the process begins and does not depend on the random choices made during the process.

We adopt some of the notation of \cite{azar1999balanced}. Each ball is assigned a fixed {\em height} upon entry, where the height is the number of balls in the bin, including itself. The height of the ball placed at time $t$ is denoted by $h(t)$. The load of a bin at time $t$ refers to the number of balls in the bin at that time. We let $\mu_{\ge k}(t)$ denote the number of balls hat have height at least $k$ at time $t$, and $\nu_{\ge k}(t)$ be the number of bins that have load at least $k$ at time $t$. Note that if a bin has load $k$, it must contain some ball of height at least $k$. Hence $\mu_{\ge k}(t) \ge \nu_{\ge k}(t)$ for all times $t$. Finally, $B(K,p)$ refers to a binomially distributed random variable based on $K$ trials each with probability $p$ of success. 

We extends the original Theorem 3.7 of \cite{azar1999balanced} and Theorem 1 of \cite{cole1998balls}, by determining a distribution on the heights of the balls that holds for polynomially many steps, regardless of which $L$ balls are in the system at any point in time.

Let $\mathcal{E}_i$ be the event that $\nu_{\ge i}(t) \le \beta_i$ for time steps $t=1,\ldots,T$, where the $\beta_i$ will be revealed shortly. We want to show that at time $t$, $1 \le t \le T$,
\begin{flalign*}
\Pr(\mu_{\ge i+1} > \beta_{i+1} | \mathcal{E}_i)
\end{flalign*}
is sufficiently small. That is, given $\mathcal{E}_i$, we want $\mathcal{E}_{i+1}$ to hold as well. This probability is hard to estimate directly. However, we know that since the $2$ choices for a ball are independent, we have 
\begin{flalign*}
\Pr(h(t) \ge i+1|\nu_{\ge i}(t-1)) = \frac{(\nu_{\ge i}(t-1))^2}{K'^2}.
\end{flalign*}
We would like to bound for each time $t$ the distribution of the number of time steps $j$ such that $h(j) \ge i+1$ and the ball inserted at time step $j$ has not been deleted by time $t$. In particular, we would like to bound this distribution by a binomial distribution over $K$ events with success probability $(\beta_i/K')^2$. But this is difficult to do directly as the events are not independent.

Instead, we fix $i$ and define the binary random variables $Y_t$ for $t=1,\ldots, T$, where 
\begin{flalign*}
Y_t=1 \ \text{iff} \ h(t) \ge i+1 \ \text{and} \ \nu_{\ge i}(t-1) \le \beta_i.
\end{flalign*}
The value $Y_t$ is 1 if and only if the height of the ball $t$ is at least $i+1$ despite the fact that the number of boxes that have load at least $i$ is currently below $\beta_i$.

Let $\omega_j$ represent the choices available to the $j$th ball. Clearly,
\begin{flalign*}
\Pr(Y_t=1|\omega_1,\ldots,\omega_{t-1},v_1,\ldots,v_{t-L}) \le \frac{\beta_i^2}{K'^2} \triangleq p_i.
\end{flalign*}
Consider the situation immediately after a time step $t'$ where a new ball has entered the system. Then there are $K$ balls in the system, that entered at times $u_1,u_2,\ldots, u_K$. Let $I(t')$ be the set of times $u_1,u_2,\ldots, u_K$. Then 
\begin{flalign*}
\sum_{t \in I(t')} Y_t = \sum_{i=1}^{K} Y_{u_i};
\end{flalign*}
that is, the summation over $I(t')$ is implicitly over the values of $Y_t$ for the balls in the system at time $t'$.

We may conclude that at any time $t' \le T$
\begin{flalign}\label{equRC:1}
\Pr\Bigg(\sum_{t \in I(t')} Y_t \ge k\Bigg) \le \Pr(B(K,p_i) \ge k).
\end{flalign}

Observe that conditioned on $\mathcal{E}_i$, we have $\mu_{\ge i+1}(t')=\sum_{t \in I(t')} Y_t$. Therefore
\begin{flalign}\label{equRC:2}
\Pr(\mu_{\ge i+1}(t') \ge k |\mathcal{E}_i )&=\Pr\Bigg(\sum_{t \in I(t')} Y_t \ge k |\mathcal{E}_i\Bigg) \\
&\le \frac{\Pr(B(K,p_i) \ge k)}{\Pr(\mathcal{E}_i)}
\end{flalign}
Thus:
\begin{flalign*}
\Pr(\neg \mathcal{E}_{i+1} |\mathcal{E}_i) \le \frac{T\Pr(B(K,p_i) \ge k) }{\Pr(\mathcal{E}_i)} 
\end{flalign*}
Since 
\begin{flalign*}
\Pr(\neg \mathcal{E}_{i+1}) \le \Pr(\neg \mathcal{E}_{i+1} |\mathcal{E}_i)  \Pr(\mathcal{E}_i) +\Pr(\neg \mathcal{E}_i),
\end{flalign*}
we have 
\begin{flalign}
\Pr(\neg \mathcal{E}_{i+1}) \le T\Pr(B(K,p_i) \ge k) + \Pr(\neg \mathcal{E}_i).
\end{flalign}
We can bound large deviations in the binomial distribution with the formula
\begin{flalign} \label{equRC:4}
\Pr(B(L,p_i) \ge ep_iK) \le e^{-p_iK}. 
\end{flalign}
We may then set $\beta_x = (K')^2/2eK$, $x=\ulcorner eK/K'\urcorner$ \cite{azar1999balanced}, and subsequently
\begin{flalign*}
\beta_{i}=eK\frac{\beta_{i-1}^2}{(K')^2} \ \text{for} \ i > x. 
\end{flalign*}
Note that the $\beta_i$ are chosen so that $\Pr(B(K,p_i) \ge \beta_{i+1}) \le e^{-p_i K}$.

With the choices $\mathcal{E}_x$, $x=\ulcorner eK/K'\urcorner$ holds \cite{azar1999balanced}, as there cannot be more than $(K')^2/2eK$ bins with $x$ balls. For $i \ge 1$,
\begin{flalign*}
\Pr(\neg \mathcal{E}_{x+i}) &\le \frac{T}{(K')^{c_1+c_2+1}} + \Pr(\neg \mathcal{E}_{x+i-1})\\
&=\frac{1}{(K')^{c_1+1}} + \Pr(\neg \mathcal{E}_{x+i-1}),
\end{flalign*}
provided that $p_iK \ge (c_1+c_2+1)\ln K'$. 

Let $i^*$ be the smallest value for which $p_{i^*-1}K \le (c_1+c_2+1) \ln K'$. Note that 
\begin{flalign}
\beta_{i+x}=\frac{K'}{2^{2^{i}}}(\frac{K'}{Ke}) \le \frac{K'}{2^{2^i}},
\end{flalign}
so $i^*=\ln \ln K'/\ln 2+ O(1)$. Then go through the standard tail bound technique in \cite{azar1999balanced} for $i \ge i^*+1$, we obtain that,
\begin{flalign*}
\Pr(\mu_{\ge x+i^*+O(c_1+c_2)} \ge 1) = O(\frac{1}{(K')^{c_1+c_2+1}}).
\end{flalign*}
So the probability that the maximum load is less than $K/K'+\ln \ln K'/\ln 2+ O(c_1+c_2)$ is bounded by
\begin{flalign*}
\Pr(\neg \mathcal{E}_{x+i^*+O(c_1+c_2)}) & \le \sum_{i=1}^{i^*}\frac{1}{(K')^{c_1+1}}+O(\frac{1}{(K')^{c_1+1}})\\
& \le O(\frac{1}{(K')^{c_1+1}}).
\end{flalign*}

\end{document}